\documentclass{ifacconf}

\usepackage{graphicx}      % include this line if your document contains figures
\usepackage{natbib}        % required for bibliography
\usepackage{amsmath} 
\usepackage{amssymb}

\usepackage{amsfonts}
\usepackage{color}
\usepackage[fleqn]{empheq}
\usepackage{enumitem}
\usepackage{mathrsfs}
\usepackage[fleqn]{empheq}
\usepackage{accents}
\usepackage{stmaryrd}
\usepackage{bbm}
\usepackage{url}
\usepackage{booktabs}

\newtheorem{theorem}{Theorem}
\newtheorem{corollary}{Corollary}

\newtheorem{definition}{Definition}
\newtheorem{remark}{Remark}
\newtheorem{problem}{Problem}
\newtheorem{proposition}{Proposition}
\newtheorem{myfact}{Fact}
\newenvironment{proof}{\textit{Proof.}}{\hfill \tiny $\blacksquare$}

\newcommand{\one}{\ensuremath{\mathbf{1}}}
\newcommand{\real}{\ensuremath{\mathbb{R}}}
\newcommand{\nat}{\ensuremath{\mathbb{N}}}

\newcommand{\bmat}[1]{\ensuremath{\begin{bmatrix}#1 \end{bmatrix}}}

\def\qedp{\hspace*{\fill}~{\tiny $\blacksquare$}}

\DeclareMathOperator{\ver}{vert}

\begin{document}
\begin{frontmatter}

\title{Data-based guarantees\\ of set invariance properties\thanksref{footnoteinfo}}

\thanks[footnoteinfo]{
This research is partially 
supported by a Marie Sk\l{}odowska-Curie COFUND grant, no.~754315.
}

\author[First]{Andrea Bisoffi}
\author[First]{Claudio De Persis} 
\author[Second]{Pietro Tesi}
		
\address[First]{ENTEG and the J.C. Willems Center for Systems and Control, University of Groningen, 9747 AG Groningen, 
The Netherlands (email: \{a.bisoffi, c.de.persis\}@rug.nl).}
\address[Second]{DINFO, University of Florence, 50139 Florence, Italy
(email: pietro.tesi@unifi.it)}

\begin{abstract}
For a discrete-time linear system, we use data from a single open-loop experiment to design directly a feedback controller enforcing that a given (polyhedral) set of the state is invariant and given (polyhedral) constraints on the control are satisfied. By building on classical results from model-based set invariance and a fundamental result from Willems et al., the controller designed from data has the following desirable features. The satisfaction of the above properties is guaranteed only from data, it can be assessed by solving a numerically-efficient linear program, and, under a certain rank condition, a data-based solution is feasible if and only if a model-based solution is feasible.
\end{abstract}

\end{frontmatter}
%===============================================================================

\section{INTRODUCTION}\label{sec:intro}

Data-driven control design is an approach that aims at 
designing control laws
based on input-output data collected from a system through an experiment.
As such, data-driven control bypasses completely the identification of a model of the plant from the input-output data. 

Auto-tuning methods (e.g., Ziegler and Nichols's method for proportional integral derivative controllers) can be seen as a seminal instance of data-driven control. 
More recent data-driven control techniques 
addressing model reference and tracking problems
include iterative feedback tuning \citep{hjalmarsson1998iterative}, 
virtual reference feedback tuning \citep{campi2002virtual}, 
iterative correlation-based tuning \citep{karimi2004iterative, Karimi2}, and
unfalsified control \citep{batt2018}.
Data-driven methods have been considered also in connection with 
other control problems, including nonlinear \citep{DFK,D2IBC},
predictive \citep{Salvador2018},
robust \citep{Dai2018}
as well as optimal control \citep{Markovsky2007,
Mukherjee2018,Baggio2019,Goncalves2019}.

Most recently, a fundamental result from~\cite{willems2005note} has been given new attention because of its deep implications for data-driven control. Namely, \cite{willems2005note} 
claims in broad terms that the whole set of solutions of a \emph{linear} system can be represented by a finite set of solutions as long as those arise from sufficiently excited dynamics. This result has been exploited 
in \cite{coulson2019data} for data-based predictive control, and in
\cite{depersis2019persistency} for data-driven stabilization and optimal control.
\cite{depersis2019persistency} shows in particular that the result by Willems 
et al. can be used to achieve a data-based parametrization of feedback systems,
enabling the design of (optimal) controllers directly via data-dependent linear matrix 
inequalities, also in the presence of noisy data. 
This idea has been further developed 
in \cite{vanwaarde2019data} to show that data-driven stabilization is possible 
even when data are not sufficiently rich to enable system identification, 
and in \cite{berberich2019robust} where -- by formulating the data-based parametrization of closed-loop systems in the presence of noisy data obtained in \cite{depersis2019persistency} as a linear fractional transformation --  data-driven $H_\infty$ control is investigated,
thus providing further evidence for developing a theory of data-driven control.

Except for contributions in the area of predictive control
such as \cite{Salvador2018} and \cite{Berberich2019mpc}, most of the works
on data-driven control do not account for state and input constraints,
which are one of the prime issues in many practical problems. 
In addition to the aforementioned papers, contributions to data-driven 
control in the presence of (state and input) constraints, also termed \emph{safe control},
are found in the literature on learning-based control \citep{Garcia2015} and
on safety certificates for learning-based control by convex optimization
\citep{wabersich2018scalable}, see also Remark~\ref{rem:wabersichZeilinger} for a detailed comparison with our approach.

In this paper, we consider data-driven safe control
using notions from set invariance \citep{blanchini1999set}.
Specifically, we consider linear time invariant (LTI) systems in discrete time, i.e.,
\begin{equation*}
x^+ = A x + B u,
\end{equation*}
and study the problem of designing a control law based on
a finite number of input-state data in such a way that the 
controlled system satisfies prescribed safety constraints,  
characterized in terms of set invariance and 
$\lambda$-contractivity (recalled in Definitions~\ref{def:(ctrl)Inv} 
and \ref{def:lambda-contractive}).

Set invariance is a dynamical property in its own right, 
and it is quite relevant as it translates the notion of safety, i.e., if
the system has initial state in this (safe) set, its solutions will not leave the set. 
Invariance is a dynamical property that is less conservative than asymptotic stability 
(e.g., for continuous-time dynamical systems, invariance of a set is essentially 
equivalent to the fact that at each point of the \emph{boundary} of the set, 
the vector field is included in the tangent cone to the set by the classical 
result in~\cite{nagumo1942lage}), but arguably as essential in practical settings. 
Ellipsoidal and polyhedral sets are common choices in the study of invariance properties, 
with the former being the level sets of classical quadratic Lyapunov functions for linear systems. 
The complexity in the representation of an ellipsoidal set contained in $\real^\nu$ is given by $\nu$, whereas that of a polyhedral set can be arbitrarily high, 
e.g., due to an arbitrarily high number of planes or vertices defining the polyhedral set. 
On the other hand, ellipsoidal sets cannot arbitrarily approximate 
any convex and compact set, whereas polyhedral sets can (see the discussion in \cite[p.~110]{blanchini2008set}).
For this reason, we consider here polyhedral sets.

Controlled invariance of polyhedral sets for discrete-time linear systems 
has been thoroughly investigated in the late 80's 
assuming exact knowledge of the matrices $A$ and $B$ above, and key results were given
\citep{gutman1986admissible,vassilaki1988feedback,blanchini1990feedback}. 
These results consider, among others, the presence of disturbances on 
the state equation and parametric uncertainties in the dynamical matrices. 
We refer the reader to the comprehensive survey \cite{blanchini1999set} 
and the monograph \cite{blanchini2008set} for an overview of these results.

Building on the notions of invariance and $\lambda$-contractivity, 
we show that the problem of designing safe controllers directly from data 
can be cast as a \emph{linear program}, which can thus be efficiently solved.
Further, as in \cite{vassilaki1988feedback,blanchini1990feedback}, 
the solution takes the form of a state-feedback gain, which avoids to
iteratively solving an online optimization problem as in receding-horizon 
predictive control and learning-based methods. On the other hand, in this paper we do not investigate optimality features of the safe controller.

The paper is organized as follows. 
Section \ref{sec:prelim} introduces the problem of interest
along with some preliminaries on set invariance.
The main results are given in Section \ref{sec:main}, while
Section \ref{sec:noise} provides a preliminary result 
in the case of noisy data. A numerical example is 
discussed in Section \ref{sec:example}, and Section \ref{sec:conclusion}
provides concluding remarks.

\textbf{Notation.} 
$\mathbb{Z}$, $\mathbb{N}$, and $\mathbb{R}$ denote the sets of integers, of nonnegative integers, and of real numbers.
For $n \in \nat$, $\nat_n := \{ 1, \dots, n\}$.
For column vectors $x_1 \in \real^{d_1}$, \dots, $x_m \in \real^{d_m}$, 
the notation $(x_1, \dots,  x_m)$ is equivalent to $[x_1^\top \dots x_m^\top]^\top$.
The $n\times n$ identity matrix is denoted by $I_n$.
The vector $\one$ denotes the vector of all ones of appropriate dimension, i.e., $\one:=(1,\dots,1)$.
Given two $n \times m$ matrices $A$ and $B$, $A \ge 0$ 
indicates that each entry of $A$ is nonnegative, and $A \ge B$ is equivalent to $A- B \ge 0$.
For a polyhedron $\mathcal{A}$, $\ver \mathcal{A}$ is the set of its vertices.
Given a set $\mathcal{A}$ and a scalar $\mu \ge 0$, $\mu \mathcal{A}:= \{ \mu x \colon x \in \mathcal{A} \}$.

\section{PROBLEM STATEMENT AND PRELIMINARIES}\label{sec:prelim}

In this section we give our problem statement and 
present essential preliminaries on set invariance.

\subsection{Problem statement}

We consider discrete-time linear time invariant (LTI) systems
\begin{equation}
\label{eq:dtLTI-alt}
x^+ = A x + B u,
\end{equation}
with state $x \in \real^n$ and input $u \in \real^m$. 
Before we introduce our sets of interest, we need the next notion.

\begin{definition}\emph{\cite[Def.~3.10]{blanchini2008set} }\label{def:C-set}
A \emph{C-set} is a convex and compact subset of 
$\real^\nu$ including the origin as an interior point.
\end{definition}

The first set of interest is the set $\mathcal{S}$ relative to the state $x$, 
which is based on a matrix $\mathrm{S} \in \real^{n_s \times n}$ with rows 
$\mathrm{S}^{(i)}$, $i=1, \dots, n_s$. 
The set $\mathcal{S}$ is a polyhedral C-set represented through $\mathrm{S}$ as
\begin{equation}
\label{eq:set S}
\begin{split}
\mathcal{S} := & \{ x \in \real^n  \colon \mathrm{S} x \le \one \} \\
=&  \{ x \in \real^n \colon \mathrm{S}^{(i)} x \le 1, i = 1, \dots, n_s \}.
\end{split}
\end{equation}
The second set of interest is the set $\mathcal{U}$ relative to the input $u$, 
which is based on a matrix $\mathrm{U} \in \real^{n_u \times m}$ 
with rows $\mathrm{U}^{(i)}$, $i=1, \dots, n_u$. 
The set $\mathcal{U}$ is a polyhedral 
convex set (including the origin as an interior point) represented through $\mathrm{U}$ as
\begin{equation}
\label{eq:set U}
\begin{split}
\mathcal{U} := & \{ u \in \real^m \colon \mathrm{U} u \le \one \} \\
= & \{ u\in \real^m \colon \mathrm{U}^{(i)} u \le 1, i = 1, \dots, n_u \}.
\end{split}
\end{equation}
We would like to impose that the 
state $x$ remains confined in the
set $\mathcal{S}$, while input $u$ is constrained in the set $\mathcal{U}$. 
To this end, we introduce the next notion of (controlled) invariance. 

\begin{definition} \emph{\cite[Defs.~4.1, 4.4]{blanchini2008set}}
\label{def:(ctrl)Inv}
A set $\mathcal{S} \subset \real^n$ is \emph{invariant} for
\begin{equation}
\label{eq:dtLTI-0}
x^+ = F x
\end{equation}
if each solution to \eqref{eq:dtLTI-0} with initial condition 
$x(0) \in \mathcal{S}$ is such that $x(t) \in \mathcal{S}$ for all $t \ge  0$.
A set $\mathcal{S} \subset \real^n$ is \emph{controlled invariant}
for
\begin{equation}
\label{eq:dtLTI}
x^+ = A x + B u
\end{equation}
if there exists a control function $u \colon \mathcal{S} \to \real^m$ 
such that for each $x(0) \in \mathcal{S}$, the corresponding solution 
to~\eqref{eq:dtLTI} satisfies $x(t) \in \mathcal{S}$ for all $t \ge  0$.
\end{definition}

We would like to impose that $\mathcal{S}$ is invariant and $u$ 
satisfies the constraints given by $\mathcal{U}$ without the knowledge 
of the matrices  $A$ and $B$, by relying only 
on a number of data samples collected from the system. 
Specifically, we make an experiment on the system by applying a
sequence $u_d(0),\dots, u_d(T-1)$ of inputs 
and measuring the corresponding values $x_d(0),\dots, x_d(T)$ of the state response, 
where the subscript $d$ emphasizes that these are data. 
Following the notation in~\cite{depersis2019persistency}, 
we organize these data as
\begin{subequations}
\label{eq:data}
\begin{align}
U_{0,T}
& :=
\begin{bmatrix}
u_d(0) & \dots & u_d(T-1)
\end{bmatrix} \label{eq:data:U0T} \\
X_{0,T}
& :=
\begin{bmatrix}
x_d(0) & \dots & x_d(T-1)
\end{bmatrix} \label{eq:data:X0T} \\
X_{1,T}
& :=
\begin{bmatrix}
x_d(1) & \dots & x_d(T)
\end{bmatrix}. \label{eq:data:X1T}
\end{align}
\end{subequations}

We can now state the problem of interest.
\begin{problem}
\label{probl:state}
Given a polyhedral C-set $\mathcal{S}$ as in~\eqref{eq:set S} and a 
polyhedral convex set $\mathcal{U}$ as in~\eqref{eq:set U}, 
find a state-feedback law $u = K x$, with gain matrix $K$ based only on the data in~\eqref{eq:data},
that \emph{guarantees} that $\mathcal{S}$ is invariant, the origin is asymptotically stable, 
and the control input $u=K x$ always belongs to $\mathcal{U}$. 
\end{problem}

For brevity, we say in the following that $\mathcal{S}$ is 
\emph{admissible for $\mathcal{U}$} if for each $x \in \mathcal{S}$, we have $K x \in \mathcal{U}$ (for some matrix $K$).

\subsection{Preliminaries on (model-based) set invariance}

In Problem~\ref{probl:state}, we ask that $\mathcal{S}$ is invariant 
and the origin is asymptotically stable. 
These two properties can be embedded in the notion of $\lambda$-contractivity defined next.

\begin{definition}\emph{\cite[Def.~4.19]{blanchini2008set}}
\label{def:lambda-contractive}
A C-set $\mathcal{S}$ is \emph{$\lambda$-contractive} for 
\begin{equation}
x^+ = F x
\end{equation}
if for some $\lambda \in [0,1)$, for each $x \in \mathcal{S}$
\begin{equation}
\inf \{ \lambda' \ge 0 \colon F x \in \lambda' \mathcal{S} \} \le \lambda.
\end{equation}

A C-set $\mathcal{S}$ is \emph{$\lambda$-contractive} for 
\begin{equation}
x^+ = A x + B u
\end{equation}
if for some $\lambda \in [0,1)$, there exists a control function 
$u \colon \mathcal{S} \to \real^m$ such that for each $x \in \mathcal{S}$
\begin{equation}
\inf \{ \lambda' \ge 0 \colon Ax + B u(x) \in \lambda' \mathcal{S} \} \le \lambda.
\end{equation}
\end{definition}

Note that if we allow $\lambda=1$ in Definition~\ref{def:lambda-contractive}, 
we recover invariance and controlled invariance of Definition~\ref{def:(ctrl)Inv} as a special case.
We recall the next result on $\lambda$-contractivity.

\begin{myfact}\emph{\cite[Thm.~4.43]{blanchini2008set}}
\label{fact:equivLambdaContr}  
Given a system
\begin{equation} \label{eq:F}
x^+ = F x
\end{equation}
and a polyhedral C-set $\mathcal{S}$ of the form~\eqref{eq:set S} 
with $\mathrm{S} \in \real^{n_s \times n}$, the set $\mathcal{S}$ is $\lambda$-contractive 
for \eqref{eq:F} if and only if there exists a matrix $P \ge 0$ such that
\begin{align}
& P \one \le \lambda \one, \\
& P \mathrm{S} = \mathrm{S} F.
\end{align}
\end{myfact}

We have the next relationship between $\lambda$-contractivity and asymptotic stability.

\begin{myfact}\emph{\cite[Cor.~4.52]{blanchini2008set}}
\label{fact:relationContrAS}
Given a system
$x^+ = F x$, there exists a polyhedral C-set which is $\lambda$-contractive 
if and only if all the eigenvalues of $F$ have modulus less or equal to 
$\lambda$ and all the eigenvalues for which the equality holds have phases that are rational multiples of $\pi$
\footnote{Namely, their phase $\theta$ can be expressed as $\theta=\frac{p}{q} \pi$ for some integers $p$ and $q$.}.
\end{myfact}

Some comments on Fact~\ref{fact:relationContrAS} are relevant for the sequel and are stated in the next remarks.

\begin{remark} \label{rem:l-contrImpliesAS}
As a consequence of Fact~\ref{fact:relationContrAS}, 
if a polyhedral C-set $\mathcal{S}$ is $\lambda$-contractive, 
then the origin (contained in the interior of $\mathcal{S}$ by 
Definition~\ref{def:C-set}) is asymptotically stable. 
Instead of imposing 
that $\mathcal{S}$ is invariant and the origin is asymptotically stable in 
Problem~\ref{probl:state}, we impose in the sequel that $\mathcal{S}$ is $\lambda$-contractive. 
Invariance of $\mathcal{S}$ ($\lambda=1$) is equivalent to marginal 
stability of the origin along with certain conditions on the eigenvalues with unitary 
modulus \citep[Thm.~4.50]{blanchini2008set}, and does not guarantee 
asymptotic stability of the origin as required by Problem~\ref{probl:state}. 
Hence imposing $\lambda < 1$ is convenient to have asymptotic stability of the origin.
\end{remark}

\begin{remark} \label{rem:ctrlAndExistence}
For state-feedback control laws $u = K x$ as in Problem~\ref{probl:state}, 
controllability of the pair $(A,B)$ implies that the 
closed-loop eigenvalues of $A+BK$ can be assigned to 
satisfy the necessary and sufficient condition in Fact \ref{fact:relationContrAS},
hence there exists a polyhedral C-set which is $\lambda$-contractive for $A+BK$.
\end{remark}

\section{Data-based design and guarantees for $\lambda$-contractivity} \label{sec:main}

We now present our data-based solution to Problem~\ref{probl:state}. By the
foregoing considerations, we address this problem in the context of 
$\lambda$-contractivity.

Given system \eqref{eq:dtLTI-alt},  $\mathcal{S}$, $\mathcal{U}$ and $u$ as in Problem~\ref{probl:state} and level of contractivity $\lambda\in [0,1)$, we have that
$\mathcal{S}$ is $\lambda$-contractive for $x^+ = (A+BK)x$
and admissible for $\mathcal{U}$ 
if and only if there exist decision variables $K$ and $P \ge 0$
such that
\begin{subequations}
\label{eq:origSol}
\begin{align}
& P \one \le \lambda \one \label{eq:lambdaContr1} \\
& P \mathrm{S} = \mathrm{S} (A+ BK) \label{eq:lambdaContr2} \\
& \mathrm{U} K s \le \one \quad \forall s \in \ver \mathcal{S}.  \label{eq:ctrlConstr}
\end{align}
\end{subequations}

Indeed, $\lambda$-contractivity of $\mathcal{S}$ is equivalent 
to~\eqref{eq:lambdaContr1}-\eqref{eq:lambdaContr2} by Fact~\ref{fact:equivLambdaContr}, 
and admissibility of $\mathcal{S}$ for $\mathcal{U}$ is equivalent to 
\begin{equation*}
K s \in \mathcal{U} \quad \forall s \in \ver \mathcal{S}
\end{equation*}
since $\mathcal{U}$ is a polyhedral convex set, and the last 
expression is equivalent to~\eqref{eq:ctrlConstr}. 
As noted in Remark~\ref{rem:l-contrImpliesAS}, a matrix $K$ that satisfies \eqref{eq:origSol} 
solves Problem~\ref{probl:state}.

We have the next result.

\begin{theorem}
\label{thm:main}
Consider $\mathcal{S}$, $\mathcal{U}$ and $u$ as in Problem~\ref{probl:state} and level of contractivity $\lambda\in [0,1)$. 
Let the data matrices $U_{0,T}$, $X_{0,T}$ and $X_{1,T}$ be as in \eqref{eq:data}.
If there exist decision variables $G_K$ and $P\ge0$ such that
\begin{subequations}
\label{eq:LP-data}
\begin{align}
& P \one \le \lambda \one \label{eq:lambdaContr-data} \\
& P \mathrm{S} = \mathrm{S} X_{1,T} G_K \\ 
& \mathrm{U} U_{0,T} G_K s \le \one 
\quad \forall s \in \ver \mathcal{S} \label{eq:ctrlConstr-data}\\
& I_n = X_{0,T} G_K,  \label{eq:consist-data} 
\end{align}
\end{subequations}
then the state-feedback gain
\begin{equation}
\label{eq:K-data}
K = U_{0,T} G_K 
\end{equation}
is such that
$\mathcal{S}$ is $\lambda$-contractive for the 
closed-loop system $x^+ = (A+BK)x$ and admissible for $\mathcal U$.
\end{theorem}

\begin{proof}
We show that the fulfilment of the constraints \eqref{eq:LP-data} implies
the fulfilment of \eqref{eq:origSol} with $K = U_{0,T} G_K$. Using \eqref{eq:K-data}
and \eqref{eq:consist-data}  we have
\begin{equation*}
\bmat{K \\ I_n} = \bmat{U_{0,T}\\ X_{0,T} } G_K,
\end{equation*}
so that
\begin{equation} \label{eq:cl_param_data}
\begin{split}
A + B K & = \bmat{B & A}\bmat{K \\ I_n}  \\
& = \bmat{B & A} \bmat{U_{0,T} \\ X_{0,T} } G_K = X_{1,T} G_K
\end{split}
\end{equation}
since $X_{1,T} = A X_{0,T} + B U_{0,T}$.
This immediately gives the claim.
\end{proof} 

\begin{remark}
We note that Theorem~\ref{thm:main} corresponds to solving 
a \emph{linear program} in the decision variables $G_K$ and $P$, 
hence it is numerically appealing.
\end{remark}

Compared with the case where the matrices $A$ and $B$ are known
(\emph{cf.} \eqref{eq:origSol}),
the data-driven solution of Theorem~\ref{thm:main} only provides 
sufficient conditions for $\lambda$-contractivity. The reason is that we made no assumptions on the data used for 
designing the controller. Intuitively, if the data do not carry enough
information on the plant dynamics, it might be impossible to
get a data-based solution.

In the context of stabilization (with no state and/or input constraints),
\cite{depersis2019persistency} shows conditions on the data which
enable a data-based parametrization of \emph{all} stabilizing state-feedback gains.
\cite{vanwaarde2019data} considers the \emph{minimum} amount of 
information on the data under which \emph{at least one} stabilizing gain
can be found from data.
Here, we follow the reasoning of \cite{depersis2019persistency},
which lends itself to a direct extension to the case of
state and/or input constraints. In fact, if the data enable 
a parametrization of \emph{all} stabilizing gains,
then any controller that guarantees $\lambda$-contractivity 
will necessarily belong to the feasibility set of \eqref{eq:LP-data} 
since $\lambda$-contractivity is a stronger property than asymptotic stability, 
as shown in Fact~\ref{fact:relationContrAS}.

The next result holds.

\begin{theorem}
\label{thm:main_v2}
Consider $\mathcal{S}$, $\mathcal{U}$ and $u$ as in Problem~\ref{probl:state}
and level of contractivity $\lambda\in [0,1)$. 
Let the data matrices $U_{0,T}$, $X_{0,T}$ and $X_{1,T}$ be as in \eqref{eq:data}.
Assume further that the matrix
\begin{equation}\label{eq:rank_cond}
\begin{split}
\Theta:= \begin{bmatrix}
U_{0,T}\\[0.1cm]
X_{0,T}
\end{bmatrix}
\end{split} 
\end{equation}
has full row rank.
Then, there exists a controller $K$ such that
$\mathcal{S}$ is $\lambda$-contractive for
$x^+ = (A+BK)x$ and admissible for $\mathcal U$
if and only if there exist decision variables $G_K$ and $P\ge0$ such that
\eqref{eq:LP-data} holds. Moreover, any such controller 
can be expressed as in \eqref{eq:K-data} for some $G_K$
satisfying \eqref{eq:LP-data}.
\end{theorem}

\begin{proof}
As justified below~\eqref{eq:origSol}, there exists a controller $K$ such that
$\mathcal{S}$ is $\lambda$-contractive for
$x^+ = (A+BK)x$ and admissible for $\mathcal U$
if and only if there exist $K$ and $P \ge 0$ such that \eqref{eq:origSol} holds. 
Theorem~\ref{thm:main} proves that if \eqref{eq:LP-data}-\eqref{eq:K-data} hold, 
then \eqref{eq:origSol} holds (sufficiency of the first part of the statement). 
On the other hand, if $\Theta$ is full row rank, the identity
\begin{equation} \label{eq:theta}
\bmat{K \\ I_n} = \Theta \, G_K 
\end{equation}
can be solved for \emph{arbitrary} $K$ with respect to $G_K$, and by the 
same derivation as in \eqref{eq:cl_param_data}, each solution $G_K$ satisfies $A + B K = X_{1,T}G_K$. 
Thus, if \eqref{eq:origSol} holds, then \eqref{eq:LP-data}-\eqref{eq:K-data} 
hold (necessity of the first part of statement). Finally, 
\eqref{eq:theta} also implies that any such controller 
can be expressed as in \eqref{eq:K-data} for some matrix $G_K$
satisfying \eqref{eq:LP-data}.
\end{proof} 

An interesting result related to the matrix $\Theta$ in~\eqref{eq:rank_cond} is that 
if the system \eqref{eq:dtLTI-alt} is controllable, then one can always ensure
that $\Theta$ has full row rank if the experimental data 
originate from exciting input signals. The result is simple 
and worth mentioning in Fact~\ref{fact:FL} below after some needed definitions.
	
\begin{definition}\label{def:Hankel}
Given a sequence $z(0),z(1),\ldots \in \mathbb{R}^\sigma$, 
we denote its \emph{Hankel matrix of depth $t$} as
\begin{equation*}
Z_{i,t,N} := \begin{bmatrix}
z(i) & z(i+1) & & \cdots & z(i+N-1)\\
z(i+1) & z(i+2) & &\cdots & z(i+N)\\
\vdots & \vdots & & \ddots & \vdots\\
z(i+t-1) & z(i+t) & & \cdots & z(i+t+N-2)\\
\end{bmatrix}
\end{equation*}
where $i \in \mathbb Z$ and $t, N \in \mathbb N$. 
For $t = 1$, we denote its Hankel matrix\footnote{
By Definition \ref{def:Hankel}, the data matrices 
$U_{0,T}$, $X_{0,T}$ and $X_{1,T}$ in~\eqref{eq:data} are precisely the 
Hankel matrices of depth $1$ of the input and (shifted) state sequences.
} as
$
Z_{i,N} := \begin{bmatrix}
z(i) & \cdots & z(i+N-1)
\end{bmatrix}
$.
\end{definition}

\smallskip
\begin{definition}\label{def:PE}
The signal $z(0),\ldots, z(T-1) \in \mathbb R^\sigma$ 
is \emph{persistently exciting of order $L$} if the matrix
$Z_{0,L,T-L+1}$ has full rank $\sigma L$.
\end{definition}

\begin{myfact} \label{fact:FL}
\emph{\cite[Cor. 2]{willems2005note}} Let system \eqref{eq:dtLTI-alt} be controllable. 
If the input sequence $u_{d}(0),\ldots,u_{d}(T-1)$ is persistently exciting 
of order $n+1$ then the 
matrix $\Theta$ has full row rank. 
\end{myfact} 

As shown in Fact \ref{fact:FL}, controllability of the system ensures that 
one can guarantee by design that $\Theta$ has full row rank.
Controllability is actually also important for enabling the existence 
of a controller achieving $\lambda$-contractivity. In fact,
for a given $\mathcal{S}$, a controller achieving $\lambda$-contractivity
need not exist. In that case, one may use the same data and search for 
different sets $\mathcal{S}'$ with different shapes until
the constraints in \eqref{eq:LP-data} become feasible. Controllability 
is beneficial in this respect because it ensures that a 
$\lambda$-contractive C-set $\mathcal{S}'$ exists, as pointed out in 
Remark~\ref{rem:ctrlAndExistence}. Alternatively, if one wants to 
design $\mathcal{S}'$, the corresponding matrix $\mathrm{S}'$ 
becomes a decision variable and \eqref{eq:LP-data} becomes 
a \emph{bilinear} program, as pointed out in \cite[p.~1755]{blanchini1999set}).

\begin{remark}
\label{rem:wabersichZeilinger}
Compared to \cite{wabersich2018scalable}, our approach considers unknown linear dynamics 
instead of known linear dynamics with unknown nonlinear term.
On the other hand, under a rank condition on the data, our approach always determines 
a solution if there is one (\emph{cf.}~Theorem~\ref{thm:main_v2}) 
instead of providing ellipsoidal under-approximations of the original polyhedral set. Moreover,
by approaching the problem in terms of
$\lambda$-contractivity, our method does not involve switching between 
a given learning-based and a designed safe controller as in \cite[Eq.~(2)]{wabersich2018scalable}, 
which may introduce undesired chattering.
\end{remark}

\subsection{$\lambda$-contractivity and decay rate}

As shown in~\cite{vassilaki1988feedback}, the function 
$V \colon \mathcal{S} \to \real$ defined as
\begin{equation}
\label{eq:V}
V(x):= \max_{i \in \{ 1, \dots , n_s\} } |\mathrm{S}^{(i)} x|
\end{equation}
is a \emph{polyhedral} Lyapunov function for the closed-loop dynamics 
$x^+ = (A + BK) x$ constrained on the set $\mathcal{S}$, 
and ensures that the origin is asymptotically stable. 
Indeed, $V$ satisfies the following properties:\newline
\textit{(i)} $V(x) \ge 0$ for all $x \in \mathcal{S}$, 
and $V(x)=0$ if and only if $x=0$ 
\footnote{
$V(x)=0$ if and only if $\mathrm{S} x = 0$, 
whose only solution is $x = 0$ because $\mathrm{S}$ has rank $n$ 
due to the C-set $\mathcal{S}$ being bounded. Indeed,
if $\mathrm{S}$ had not rank $n$, the nullspace of $\mathrm{S}$ 
would contain at least one nonzero vector $\bar x \neq 0$ satisfying $\mathrm{S} \bar x = 0$, 
so that $\mathrm{S} (M \bar x) = 0$ would also hold for an arbitrarily large $M$. 
But then $\mathrm{S} (M \bar x)  \le \one$ and $M \bar x \in \mathcal{S}$, 
contradicting boundedness of $\mathcal{S}$.},\newline
\textit{(ii)} it holds that
\begin{align}
& V(x^+) := \max_{i \in \{ 1, \dots , n_s\} } |\mathrm{S}^{(i)} x^+| =  
\max_{i \in \{ 1, \dots , n_s\} } |\mathrm{S}^{(i)} (A+ BK) x| \nonumber \\
& \overset{\eqref{eq:lambdaContr2}}{=} \max_{i \in \{ 1, \dots , n_s\} } 
\bigg| \sum_{j=1}^n p_{ij} \mathrm{S}^{(j)} x \bigg| \le \max_{i \in \{ 1, \dots , n_s\} } 
\sum_{j=1}^n | p_{ij} | |\mathrm{S}^{(j)} x| \nonumber \\
& \overset{\eqref{eq:V}}{\le} \max_{i \in \{ 1, \dots , n_s\} }  
V(x) \sum_{j=1}^n | p_{ij} | \overset{P \ge 0,\,\eqref{eq:lambdaContr1} }{\le} \lambda V(x). \label{eq:V(x+)}
\end{align}

Properties \textit{(i)} and \textit{(ii)} imply asymptotic stability of the origin. 
In view of~\eqref{eq:V(x+)}, the level of contractivity $\lambda$ is also 
the decay rate of the Lyapunov function $V$, and it is thus of interest to 
minimize $\lambda \in [0,1)$ as proposed for instance in~\cite{vassilaki1988feedback}. 
It is straightforward to do this based \emph{only} on data, as shown in the next result.

\begin{corollary}
\label{cor:main}
Consider the same setting as in Theorem \ref{thm:main}.
If there exist decision variables $\lambda$, $G_K$ and $P \ge 0$ solving
\begin{equation}
\label{eq:LP-min-data} 
\begin{split}
& \min \lambda \\
& \text{ \emph{such that} } 0 \le \lambda < 1 \text{ \emph{and} } \eqref{eq:LP-data} \text{ \emph{holds,}}
\end{split}
\end{equation}
the controller $K$ as in \eqref{eq:K-data} ensures that
$\mathcal{S}$ is $\lambda$-contractive for $x^+ = (A+BK)x$ 
and admissible for $\mathcal{U}$. 
\qedp
\end{corollary}

The decision variables $\lambda$, $G_K$ and $P$ enter~\eqref{eq:LP-min-data} 
in a linear fashion. Hence, \eqref{eq:LP-min-data} still
corresponds to a linear program and can then be solved efficiently.

\section{ROBUST DESIGN FOR NOISY DATA} \label{sec:noise}

In this section we present some preliminary result for the more realistic setting of noisy data. To this end, we consider a system of the form
\begin{equation}
\label{eq:sysDist}
x^+ = A x + B u + d,
\end{equation}
where $d \in \mathcal{D} \subset \real^n$ and $\mathcal{D}$ is a polyhedral C-set represented through convex combinations of its $n_d$ vertices $d^{(1)}, \dots, d^{(n_d)} \in \real^n$ as
\begin{equation}
\label{eq:set D}
\begin{split}
\mathcal{D} := & \Bigg\{ \sum_{i=1}^{n_d} \alpha_i d^{(i)} \colon \one^\top \alpha =1, \alpha \ge 0 \Bigg\}.
\end{split}
\end{equation} 
The disturbance affects both the data and the invariance properties of~\eqref{eq:sysDist}. As for the data, the experiment involves the quantities in~\eqref{eq:data} and, additionally, the \emph{unknown} sequence $d_d(0), \dots, d_d(T-1)$ of disturbances, organized as
\begin{equation}
\label{eq:data:D0T}
D_{0,T} :=
\begin{bmatrix}
d_d(0) & \dots & d_d(T-1)
\end{bmatrix}. 
\end{equation}
The overall data in~\eqref{eq:data:D0T} and \eqref{eq:data} satisfy then from \eqref{eq:sysDist} that
\begin{equation}
\label{eq:sysDynStack}
\begin{split}
X_{1,T} & = A X_{0,T} + B U_{0,T} + D_{0,T}\\
&  = \bmat{B & A}  \bmat{U_{0,T} \\ X_{0,T} } + D_{0,T}.
\end{split}
\end{equation}
As for the invariance properties, we consider accordingly the next robust version of Definition~\ref{def:(ctrl)Inv}.
\begin{definition}\emph{\cite[Def.~2.1]{blanchini1990feedback}}
A set $\mathcal{S}$ is \emph{robustly invariant with respect to $\mathcal{D}$ for}
\begin{equation}
\label{eq:sysCLdist}
x^+ = F x + d
\end{equation}
if for each initial condition $x(0) \in \mathcal{S}$ and each disturbance $d$ satisfying $d(t) \in \mathcal{D}$ for all $t \ge 0$, the corresponding solution to~\eqref{eq:sysCLdist} satisfies $x(t) \in \mathcal{S}$ for all $t \ge 0$.
\end{definition}

In this section we consider a slightly different setting than the rest of the paper, that is, guaranteeing that $\mathcal{S}$ is robustly invariant w.r.t. $\mathcal{D}$ for the closed-loop system and is admissible for $\mathcal{U}$, in the presence of noisy data. We recall the next instrumental result.
\begin{myfact}\emph{\cite[Thm.~2.1]{blanchini1990feedback}} \label{fact:robust inv}
Let $\mathcal{S}$ and $\mathcal{D}$ be C-sets. The set $\mathcal{S}$ is robustly invariant w.r.t. $\mathcal{D}$ for~\eqref{eq:sysCLdist} if and only if for each $s \in \ver \mathcal{S}$ and each $w \in \ver \mathcal{D}$, $F s +  w \in \mathcal{S}$.
\end{myfact}
This fact allows us to conclude that given the system in~\eqref{eq:sysDist} and for $\mathcal{S}$ and $\mathcal{U}$ and $u$ as in Problem~\ref{probl:state} 
and the C-set $\mathcal{D}$ in~\eqref{eq:set D},
$\mathcal{S}$ is
\begin{enumerate}[label=\textit{(\alph*)}]
\item robustly invariant w.r.t. $\mathcal{D}$ for $x^+ = (A+BK) x + d$,
\item admissible for $\mathcal{U}$
\end{enumerate}
if and only if there exists a decision variable $K$ such that
\begin{subequations}
\label{eq:iff robust inv}
\begin{align}
& \mathrm{S} ( (A+BK) s + w) \le \one & & \forall s \in \ver \mathcal{S}, \forall w \in \ver \mathcal{D}\\
& \mathrm{U} K  s \le \one  & & \forall s \in \ver \mathcal{S}.
\end{align}
\end{subequations} 

Let us apply to~\eqref{eq:iff robust inv} the same approach as in Section~\ref{sec:main} in light of the new dynamics in~\eqref{eq:sysDynStack}. If there exists a decision variable $G_K$ such that
\begin{subequations}
\label{eq:if robust inv data}
\begin{align}
& \mathrm{S} ( (X_{1,T} - D_{0,T} ) G_K s + w) \le \one & & \forall s \in \ver \mathcal{S}, \forall w \in \ver \mathcal{D} \label{eq:if robust inv data:inv}\\
& \mathrm{U} U_{0,T} G_K  s \le \one  & & \forall s \in \ver \mathcal{S}\\
& I_n = X_{0,T} G_K, & & 
\end{align}
\end{subequations} 
then the state-feedback gain $K = U_{0,T} G_K$ would ensure for $\mathcal{S}$ its desired properties \textit{(a)}--\textit{(b)} above. 
In particular, \eqref{eq:if robust inv data:inv} follows from
\begin{equation*}
\begin{split}
A + B K & = \bmat{B & A}\bmat{K \\ I_n}  \\
& = \bmat{B & A} \bmat{U_{0,T} \\ X_{0,T} } G_K =( X_{1,T} - D_{0,T} )G_K
\end{split}
\end{equation*}
where the last equality uses the new dynamics in~\eqref{eq:sysDynStack}.
However, the disturbance sequence leading to $D_{0,T}$ in \eqref{eq:if robust inv data:inv} is unknown. A possible way of overcoming this issue is to ask conservatively that \eqref{eq:if robust inv data:inv} be satisfied for all the possible sequences of the disturbance $d_d(0), \dots, d_d(T-1)$ as long as each $d_d(0), \dots, d_d(T-1)$ belongs to $\mathcal{D}$. To this end, define for $j \in \nat_T$ and $i \in \nat_{n_d}$ the matrix $\delta_{ji} \in \real^{n \times T}$ being zero except for its $j$-th column equal to $T d^{(i)}$, i.e.,
\begin{align*}
\begin{minipage}{1.1cm}
$\delta_{ji} : =$ \vspace*{.1cm}\\
\end{minipage} 
& \begin{matrix}  \big[ \underbrace{0}_{1\text{-st},} | & \dots &  | \underbrace{ T d^{(i)}}_{j\text{-th},} | &  \dots & |  \underbrace{0}_{T\text{-th column}} \big].  \end{matrix}
\end{align*}
The reason for the dependence on $T$ in the $j$-th column of $\delta_{ji}$ becomes clear in the proof of our next result.
\begin{proposition}
\label{propo:robust}
Consider $\mathcal{S}$, $\mathcal{U}$ and $u$ as in Problem~\ref{probl:state}, the disturbance $d$ belonging to the C-set $\mathcal{D}$ in~\eqref{eq:set D}, and let the data matrices $U_{0,T}$, $X_{0,T}$, $X_{1,T}$ and $D_{0,T}$ be as in \eqref{eq:data} and \eqref{eq:data:D0T}.
If there exists a decision variable $G_K$ such that%
\begin{subequations}
\label{eq:LP robust inv}
\begin{align}
& \mathrm{S} ( (X_{1,T} - \delta_{ji} ) G_K s + w) \le \one  \nonumber\\
& \hspace*{1cm} \forall s \in \ver \mathcal{S}, \forall w \in \ver \mathcal{D}, \forall j \in \nat_T, \forall i \in \nat_{n_d}  \label{eq:LP robust inv:inv}\\
& \mathrm{U} U_{0,T} G_K  s \le \one  \hspace*{1cm} \forall s \in \ver \mathcal{S}\\
& I_n = X_{0,T} G_K, 
\end{align}
\end{subequations}
then the state-feedback gain
\begin{equation*}
K = U_{0,T} G_K 
\end{equation*}
is such that $\mathcal{S}$ is robustly invariant w.r.t. $\mathcal{D}$ for $x^+ = (A+BK)x + d$ and admissible for $\mathcal U$.
\end{proposition}
\begin{proof}
The statement is proven if we show that \eqref{eq:LP robust inv:inv} implies that \eqref{eq:if robust inv data:inv} is verified for \emph{all} possible $D_{0,T}$ because \eqref{eq:if robust inv data} implies that \eqref{eq:iff robust inv} holds (with the same arguments as in the proof of Theorem~\ref{thm:main}), and \eqref{eq:iff robust inv} guarantees the statement.

Since each column of $D_{0,T}$ belongs to the C-set $\mathcal{D}$ in~\eqref{eq:set D}, $D_{0,T}$ can be written as
\begin{equation*}
D_{0,T}= \bmat{\sum_{i=1}^{n_d} \alpha_{1,i} d^{(i)} \Big| & \dots & \Big| \sum_{i=1}^{n_d} \alpha_{T,i} d^{(i)} }
\end{equation*}
where the vectors $\alpha_1, \dots, \alpha_T$ satisfy $\one^\top \alpha_1 = 1$ and $\alpha_1 \ge 0$, \dots, $\one^\top \alpha_T = 1$ and $\alpha_T \ge 0$. 
Consider in the rest of the proof arbitrary  $s \in \ver \mathcal{S}$ and $w \in \ver \mathcal{D}$.
\eqref{eq:LP robust inv:inv} implies that for such $s$ and $w$ and for each $j \in \nat_T$ and $i \in \nat_{n_d}$, 
\begin{equation}
\label{eq:conseq 1 of for all}
\frac{\alpha_{j,i} }{T} \mathrm{S} X_{1,T} G_K s 
- \frac{\alpha_{j,i} }{T} \mathrm{S} \delta_{ji} G_K s +
\frac{\alpha_{j,i} }{T} \mathrm{S}  w
\le \frac{\alpha_{j,i} }{T} \one
\end{equation}
because each $\alpha_{j,i} \ge 0$.
\eqref{eq:conseq 1 of for all} implies that the summation over all $i \in \nat_{n_d}$ holds as well, i.e., for each $j \in \nat_T$
\begin{align}
\label{eq:conseq 2 of for all}
& \frac{1}{T} \mathrm{S} X_{1,T} G_K s 
- \frac{1}{T} \mathrm{S} 
\bmat{0 \Big| & \dots & \Big| \sum_{i=1}^{n_d} \alpha_{j,i} T d^{(i)} \Big| & \dots & \Big|0} G_K s \nonumber \\
& \hspace*{5cm}+ \frac{1}{T} \mathrm{S}  w
\le \frac{1}{T} \one.
\end{align}
\eqref{eq:conseq 2 of for all} implies that the summation over all $j \in \nat_T$ holds as well, i.e.,
\begin{align*}
& \mathrm{S} X_{1,T} G_K s 
- \mathrm{S} 
\bmat{\sum_{i=1}^{n_d} \alpha_{1,i}  d^{(i)} \Big| & \dots & \Big|  \sum_{i=1}^{n_d} \alpha_{T,i}  d^{(i)} } G_K s \\
& \hspace*{4.7cm} + \mathrm{S}  w
\le  \one
\end{align*}
so \eqref{eq:if robust inv data:inv} holds indeed, which was the implication needed to complete the proof.
\end{proof}

Proposition~\ref{propo:robust} is a preliminary result due to the conservatism of replacing the constraints in~\eqref{eq:if robust inv data:inv} (where $D_{0,T}$ is unknown) with $n_d T$ as many such constraints in~\eqref{eq:LP robust inv:inv}. On the other hand, Proposition~\ref{propo:robust} still corresponds to solving a linear program in the decision variable $G_K$.

\section{NUMERICAL EXAMPLE}\label{sec:example}

In this section we illustrate the results of Section~\ref{sec:main} through an example taken from~\cite{vassilaki1988feedback}. 

The sets $\mathcal{S}$ in~\eqref{eq:set S} and $\mathcal{U}$ in~\eqref{eq:set U} 
are determined by the next matrices $\mathrm{S}$ and $\mathrm{U}$:
\begin{equation}
\label{eq:matrices S and U example}
\mathrm{S}:=
\begin{bmatrix*}[r]
{1}/{5} & & {2}/{5} \\
-{1}/{5} & & -{2}/{5} \\
-{3}/{20} & & {1}/{5} \\
{3}/{20} & & -{1}/{5} \\
\end{bmatrix*},\quad
\mathrm{U}:=
\begin{bmatrix*}[r]
1/7\\
-1/7
\end{bmatrix*},
\end{equation}
so that the set $\mathcal{S}$ corresponds to the
quadrilateral in a green, solid line in Figure~\ref{fig:sim}, 
while the set $\mathcal{U}$ corresponds to the condition $-7 \le u \le 7$. 
The level of contractivity is selected as $\lambda=0.84$.

\begin{figure}
\centering
\includegraphics[width=.95\columnwidth]{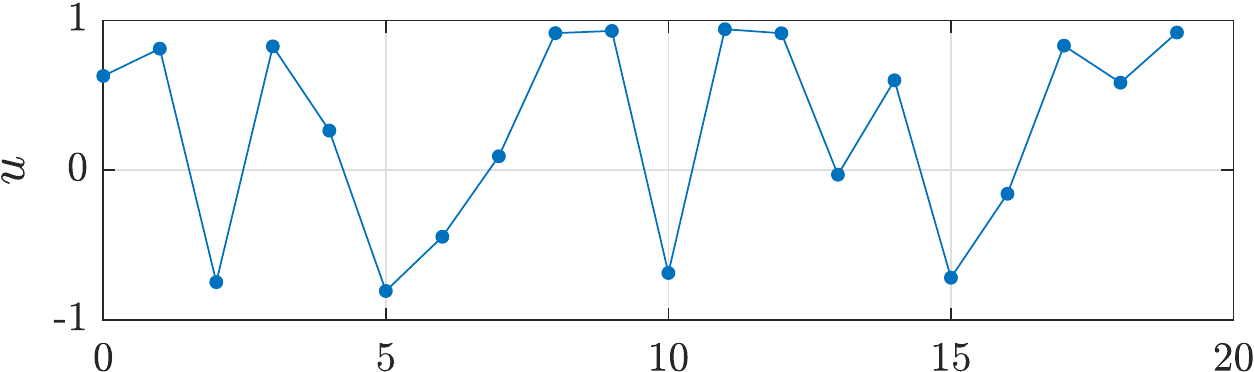}
\includegraphics[width=.95\columnwidth]{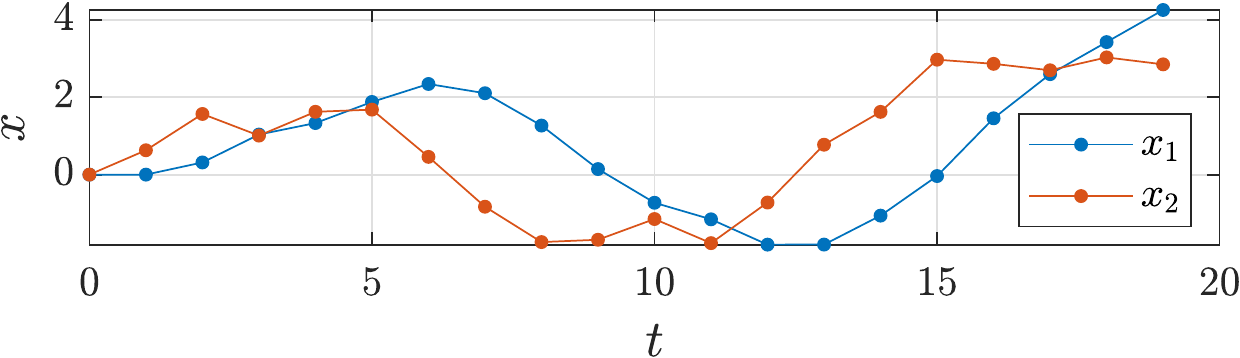}
\caption{Input and state sequences of data as in~\eqref{eq:data}, with $T=20$.}
\label{fig:data} 
\end{figure}

The data are collected from an open-loop experiment as in Figure~\ref{fig:data}, 
where $u$ is the realization of a random variable uniformly distributed on $[-1,1]$, 
and show that the underlying linear system is unstable.
The matrices $A$ and $B$ generating these data are
\begin{equation}
A:=
\begin{bmatrix*}[r]
4/5 & & 1/2\\
-2/5 & & 6/5
\end{bmatrix*}, \quad 
B:=
\begin{bmatrix}
0\\
1
\end{bmatrix},
\end{equation}
and are reported \emph{only} for illustrative purposes, because our solution 
relies \emph{only} on the collected data, as per Theorem~\ref{thm:main}. 

\begin{remark}
Full row rank of $\Theta$ in~\eqref{eq:rank_cond} can be checked from data. 
However, this condition holds by Fact~\ref{fact:FL} if $(A,B)$ is controllable 
and the input sequence is persistently exciting of order $n+1$.  As noted in 
\cite[\S~II.A]{depersis2019persistency}, persistence of excitation (see Definition~\ref{def:PE}) 
poses a mild necessary condition on the number of samples, i.e., $T \ge (m+1) n + m = 5$ in the considered case.
\end{remark}

\begin{figure}
\centering
\includegraphics[width=.95\columnwidth]{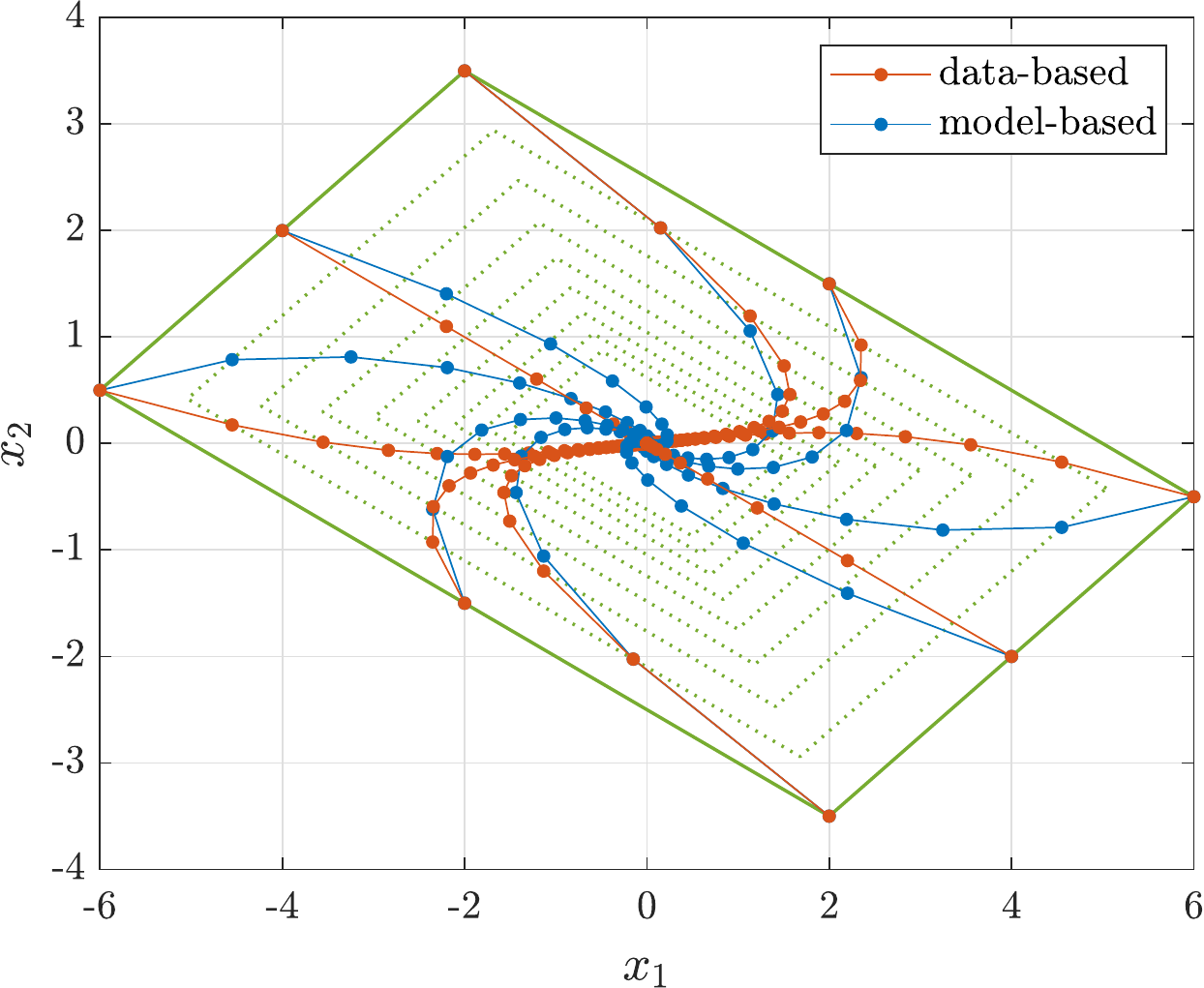}
\includegraphics[width=.95\columnwidth]{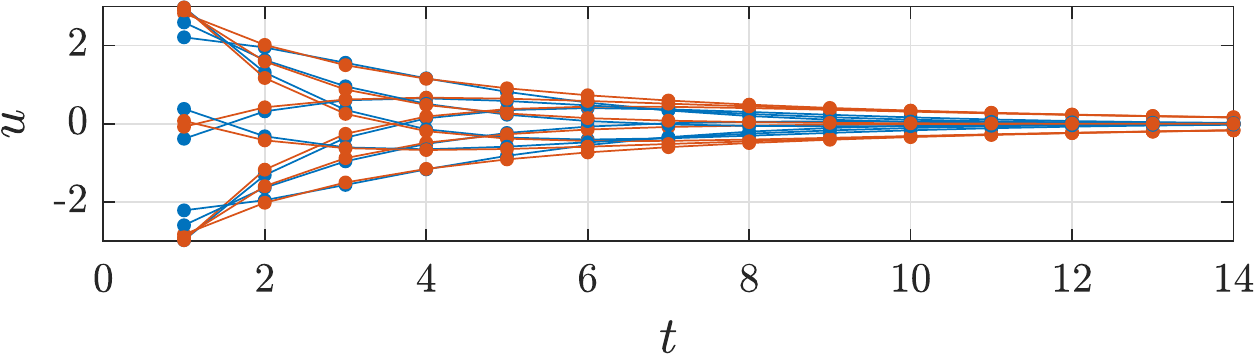}
\caption{Sets $\mathcal{S}$ and $\mathcal{U}$ given as in~\eqref{eq:matrices S and U example} 
and $\lambda = 0.84$. (Top) Solutions arising from the state feedback law $u=K x$ (see \eqref{eq:K}) 
designed based on data (orange), and from $u=K_{A,B} x$ (see \eqref{eq:K_AB}) based on the classical 
model-based approach (blue), set $\mathcal{S}$ (green, solid) and the sets $\lambda \mathcal{S}$, 
$\lambda^2 \mathcal{S}$, $\lambda^3 \mathcal{S}$, \dots (green, dotted). (Bottom) 
Control signal $u$ corresponding to the solutions in orange and blue depicted on top. 
The control signal  satisfies the constraints given by $\mathcal{U}$.}
\label{fig:sim} 
\end{figure}

The linear optimization problem in Theorem~\ref{thm:main} is solved in the variables 
$G_K$ and $P$, and the resulting $K$ in~\eqref{eq:K-data} is
\begin{equation}
\label{eq:K}
K=\bmat{0.420  &  -0.610}.
\end{equation}
{Only} for illustrative purposes, we also solve the problem in~\eqref{eq:origSol} and obtain a gain matrix
\begin{equation}
\label{eq:K_AB}
K_{A,B}=\bmat{0.313  &  -0.671}.
\end{equation}
The solutions resulting from simulating the system with state feedback law $u=K x$ 
(our data-based solution) and $u=K_{A,B} x$ (the model-based solution) are in 
Figure~\ref{fig:sim} and show that Problem~\ref{probl:state} is solved.

As an alternative to solving the feasibility problem in Theorem~\ref{thm:main}, 
we solve the minimization problem in Corollary~\ref{cor:main} using the same data. 
In this case we obtain $\lambda = 0.758$ and $K=K_{A,B}=\bmat{0.379  &  -0.692}$ and the resulting solutions are in Fig.~\ref{fig:sim-min}. 

Some comments on the results corresponding to Figures~\ref{fig:sim} and \ref{fig:sim-min} 
can be made. Because $\Theta$ in~\eqref{eq:rank_cond} has full row rank, feasibility of conditions~\eqref{eq:origSol} 
in the variables $K$ and $P$ is equivalent to feasibility of conditions \eqref{eq:LP-data} in the 
variables $G_K$ and $P$ by Theorem~\ref{thm:main_v2}. In general, the two feasibility problems 
yield different solutions as in Figure~\ref{fig:sim}, e.g., due to different initializations of the decision variables. 
However, since feasible linear programs have a global minimum, minimizing $\lambda$ 
under \eqref{eq:origSol} or \eqref{eq:LP-data} yields the same value for $\lambda$. 
Moreover, minimizing $\lambda$ reduces the size of the feasibility set (due to the constraints $P \ge 0$ and $P \one \le \lambda \one$),
which leads in this case to the fact that the minimizers $G_K$ and $P$ under the conditions in~\eqref{eq:LP-data} 
yield the same feedback gain as the minimizers $K$ and $P$ under the conditions in~\eqref{eq:origSol}.

\begin{figure}
\centering
\includegraphics[width=.95\columnwidth]{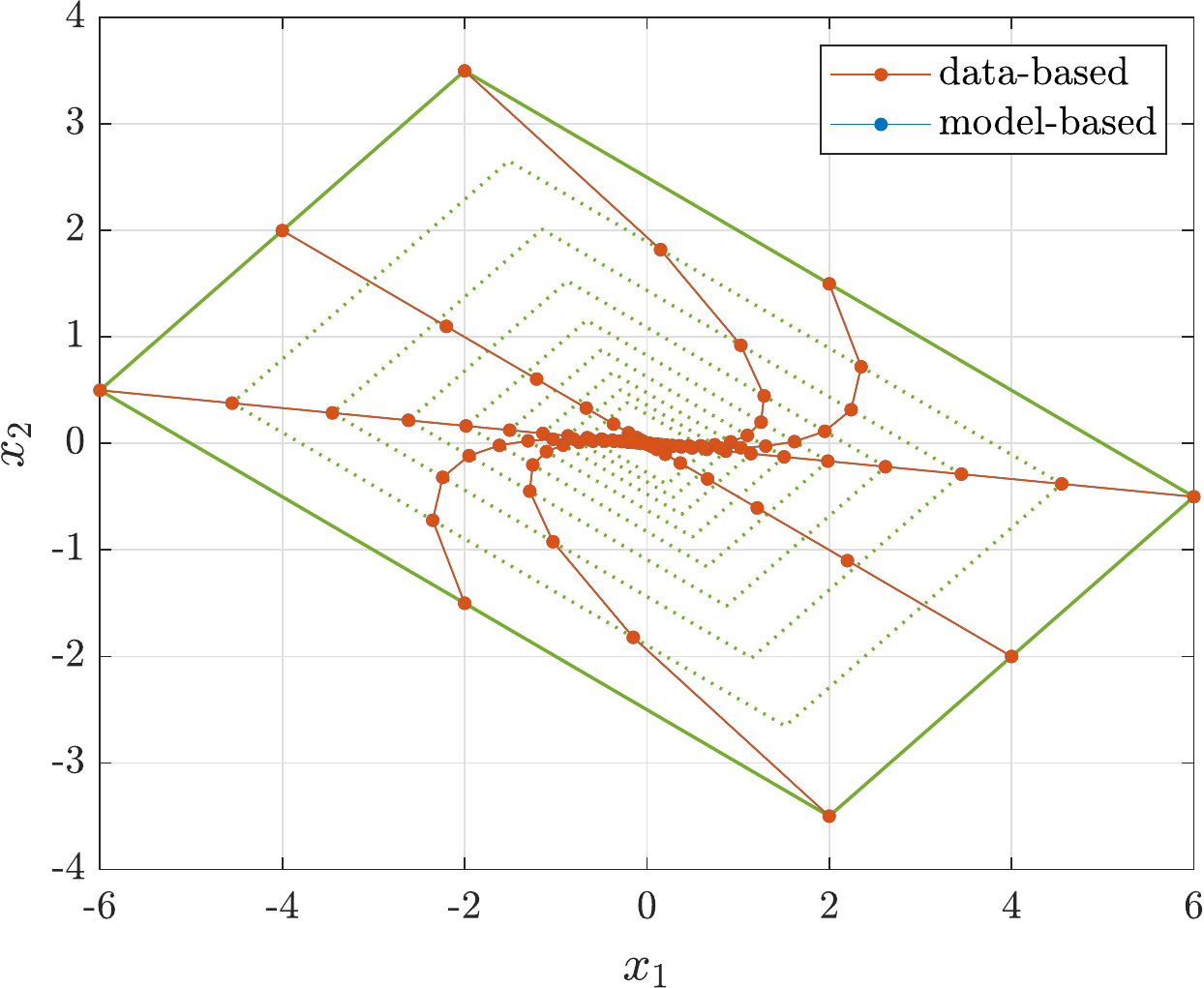}
\includegraphics[width=.95\columnwidth]{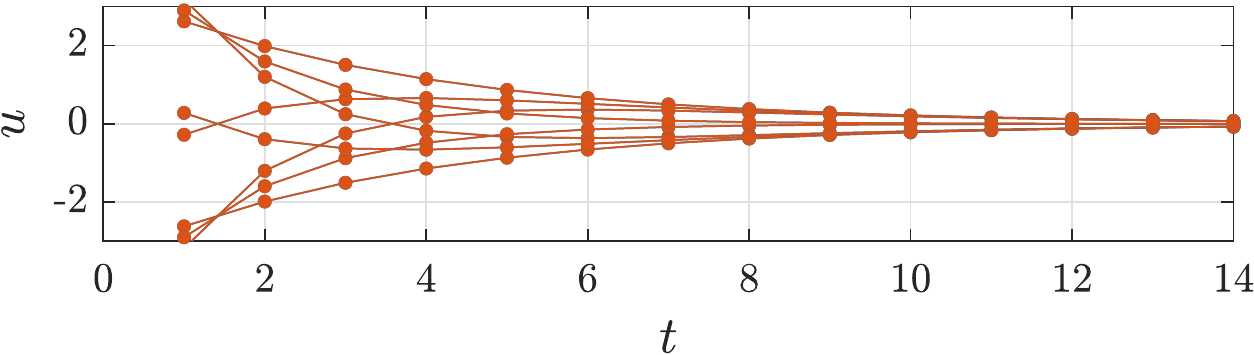}
\caption{See the caption of Figure~\ref{fig:sim} for the illustration convention 
of the quantities in this figure, which correspond to $\lambda=0.758$ minimized as in Corollary~\ref{cor:main}.}
\label{fig:sim-min} 
\end{figure}

\section{CONCLUSIONS} \label{sec:conclusion}

This paper proposes a data-based solution for designing a controller enforcing that a given polyhedral C-set for the state is $\lambda$-contractive (hence, invariant) and given polyhedral convex constraints on the control are satisfied. With respect to classical approaches from set-invariance, we show that the data-based solution still arises from a numerically-efficient linear program, and that, under a rank condition on the collected data, the data-based solution is feasible if and only if the model-based solution is feasible. The level of $\lambda$-contractivity is guaranteed based on the data.
Our main results are given for the nominal case when the input and state data are not affected by noise and a preliminary result is given for noisy data. 

\bibliography{refs}

\end{document}